\documentclass[a4paper,UKenglish,cleveref, autoref, thm-restate, review]{lipics-v2021}

\title{Improved Ackermannian lower bound for the Petri nets reachability problem} 

\author{Sławomir Lasota}{University of Warsaw, Poland}{}{https://orcid.org/0000-0001-8674-4470}
{Supported by the  ERC project ‘Lipa’ within the EU Horizon 2020 research and innovation programme (No. 683080).}
\authorrunning{S. Lasota} 
\Copyright{Sławomir Lasota}



\category{} 

\relatedversion{} 

\nolinenumbers 

\EventEditors{Petra Berenbrink and Benjamin Monmege}
\EventNoEds{2}
\EventLongTitle{39th International Symposium on Theoretical Aspects of Computer Science (STACS 2022)}
\EventShortTitle{STACS 2022}
\EventAcronym{STACS}
\EventYear{2022}
\EventDate{March 15--18, 2022}
\EventLocation{Marseille, France}
\EventLogo{}
\SeriesVolume{219}
\ArticleNo{35}

\usepackage{todonotes}
\usepackage{algorithm}
\usepackage[noend]{algpseudocode}
\usepackage{xspace}
\usepackage{pgf}
\usepackage{tikz}
\usetikzlibrary{arrows,automata}

\begin{document}

\maketitle



\newcommand{\computed}[3]{\text{\sc Comp}_{#1}(#2, #3)}
\newcommand{\twocol}[5]{
\begin{minipage}{#1\linewidth} #2 \end{minipage}
#3
\begin{minipage}{#4\linewidth} #5 \end{minipage}
}
\newcommand{\moveup}{\vspace{-1mm}}
\newcommand{\movedown}{\vspace{1mm}}

\newcommand{\finval}[1]{\textsc{fin}(#1)}
\newcommand{\runs}[3]{\textsc{Runs}_{#1}(#2, #3)}
\newcommand{\true}{\mathbf{true}}
\newcommand{\false}{\mathbf{false}}
\newcommand{\NP}{\text{\sc NP}}
\newcommand{\divs}[1]{\text{\sc Div}(#1)}
\newcommand{\setof}[2]{\{ #1 \mid #2 \}}
\newcommand{\set}[1]{\{#1\}}
\newcommand{\lcmpar}[1]{\text{\sc Lcm}(#1)}
\newcommand{\bin}[1]{\text{\sc Bin}(#1)}
\newcommand{\binext}[2]{\text{\sc Bin}_{#1}(#2)}
\newcommand{\size}[1]{|#1|}
\newcommand{\I}{\mathcal{I}}
\renewcommand{\O}{\mathcal{O}}
\newcommand{\N}{\mathbb{N}}
\newcommand{\Npar}[1]{\mathbb{N}_{#1}}
\newcommand{\Ne}{\Npar 4}
\newcommand{\Z}{\mathbb{Z}}
\newcommand{\R}{\mathbb{R}}
\newcommand{\Q}{\mathbb{Q}}

\newcommand{\lcm}{\textsc{LCM}(1,\ldots,n)}
\newcommand{\weakexp}[3]{\mathcal{W}_{#1}}
\newcommand{\commanda}[1]{\textbf{command}_{#1}}
\newcommand{\commandb}[1]{\textbf{command2}_{#1}}


\newcommand{\ptime}{\textsc{PTime}\xspace}
\newcommand{\pspace}{\textsc{PSpace}\xspace}
\newcommand{\nl}{\textsc{NL}\xspace}
\newcommand{\np}{\textsc{NP}\xspace}
\newcommand{\expspace}{\textsc{ExpSpace}\xspace}
\newcommand{\exptime}{\textsc{ExpTime}\xspace}
\newcommand{\nexptime}{\textsc{NExpTime}\xspace}
\newcommand{\tower}{\textsc{Tower}\xspace}
\newcommand{\ackermann}{\textsc{Ackermann}\xspace}

\newcommand{\trans}[1]{\xrightarrow{#1}}
\newcommand{\tran}{\longrightarrow}
\newcommand{\Tran}{\Longrightarrow}

\newcommand{\slawek}[1]{\todo[inline, color=orange!50]{{\bf SL:} #1}}

\newcommand{\hp}{\mathcal{HP}}
\newcommand{\subsum}{\textsc{Subset Sum}\xspace}



\newcommand{\inc}[1]{\add{#1}{1}}
\newcommand{\dec}[1]{\sub{#1}{1}}
\newcommand{\incordec}[1]{$#1 \,\, *\!\!= \, 1$}
\newcommand{\initialise}{\textbf{initialise to} $0$}
\newcommand{\coreadd}[2]{#1 \,\, +\!\!= \, #2}
\newcommand{\coresub}[2]{#1 \,\, -\!\!= \, #2}
\newcommand{\add}[2]{$\coreadd{#1}{#2}$}
\newcommand{\sub}[2]{$\coresub{#1}{#2}$}
\newcommand{\mul}[2]{$#1 \,\, \times\!\!= \, #2$}
\newcommand{\epar}[1]{e_{k, #1}}
\newcommand{\enopar}{e_k}
\newcommand{\para}[1]{\subparagraph*{#1}}
\makeatletter
\def\BState{\State\hskip-\ALG@thistlm}
\makeatother
\newcommand{\Loopexactly}[1]{\Loop\ \textbf{exactly $#1$ times}}
\newcommand{\Loopatmost}[1]{\Loop\ \textbf{at most $#1$ times}}
\newcommand{\Loopdown}[1]{\Loop\ \textbf{down $#1$}}
\newcommand{\op}{\text{op}}
\newcommand{\comp}[2]{#1\; #2}
\newcommand{\compn}[3]{#2 \rhd_{#1} #3}
\newcommand{\segm}[1]{[#1]}
\newcommand{\loopdown}[1]{\textbf{loop down} #1}

\renewcommand{\thealgorithm}{\Roman{algorithm}}
\algrenewcommand{\algorithmiccomment}[1]{\qquad$\rightarrow$ #1}
\newcommand{\goto}[2]{\textbf{goto} {\footnotesize #1} \textbf{or} {\footnotesize #2}}
\newcommand{\gotod}[1]{\textbf{goto} {\footnotesize #1}}
\newcommand{\testz}[1]{\textbf{zero?}~$#1$}
\newcommand{\testm}[1]{\textbf{max?}~$#1$}
\newcommand{\halt}{\textbf{halt}}
\newcommand{\haltz}[1]{{\halt} \textbf{if} $#1 = 0$}
\newcommand{\startz}[1]{\textbf{start with} $#1 = 0$}
\newcommand{\vr}[1]{\mathsf{#1}}
\newcommand{\tuple}[1]{\langle #1 \rangle}
\newcommand{\ass}[2]{[#1 := #2]}

\newcommand{\iszero}[1]{\textbf{iszero} $#1$}
\newcommand{\reset}[1]{\textbf{reset} $#1$}

\newcommand*{\eg}{e.g.\@\xspace}
\newcommand*{\ie}{i.e.\@\xspace}

\newcommand{\zerotestnew}[4]{\zerotestname\,#1} 
\newcommand{\zerotestname}{\text{\sc Zero?}}
\newcommand{\zerotest}[1]{\zerotestname(#1)}
\newcommand{\settozero}{\text{\sc Set-$\vr c$-to-zero}}
\newcommand{\PROG}[4]{
\begin{minipage}{#1\linewidth}
\medskip
\underline{\bf #2:}
\label{#3}
\begin{algorithmic}[1]
#4
\end{algorithmic}
\end{minipage}
}
\newcommand{\PROGnoname}[3]{
\begin{minipage}{#1\linewidth}
\label{#2}
\begin{algorithmic}[1]
#3
\end{algorithmic}
\end{minipage}
}
\newcommand{\prog}[1]{\mathcal{#1}}
\newcommand{\successor}[1]{\widetilde #1}
\newcommand{\FF}[1]{\mathcal{F}_{#1}}
\newcommand{\FFfun}[1]{\mathcal{FF}_{#1}}
\newcommand{\F}[1]{\mathbf{F}_{#1}}
\newcommand{\Ack}[1]{\mathbf{A}_{#1}}
\newcommand{\ini}[1]{\overline{\vr #1}}
\newcommand{\fin}[1]{\underline{\vr #1}}
\newcommand{\midd}[1]{\dot{\vr #1}}
\newcommand{\middd}[1]{\vr #1{\cdot}{\cdot}}
\newcommand{\duch}[1]{\vr #1^*}
\newcommand{\iniduch}[1]{\overline{\duch #1}}
\newcommand{\finduch}[1]{\underline{\duch #1}}
\newcommand{\transf}[1]{#1^*\!}
\newcommand\overoverline[1]{\widehat{#1}}  
\newcommand{\rel}[1]{R_{#1}}
\newcommand{\prettyexists}[2]{\exists #1 : #2}
\newcommand{\prettyforall}[2]{\forall #1 : #2}
\newcommand{\ratio}[1]{\text{\sc Ratio}_{}}
\newcommand{\zeroval}{{\bf 0}}
\newcommand{\val}{\vec v}
\newcommand{\linampl}[1]{\prog L_{#1}}
\newcommand{\mult}[1]{\prog M_{#1}}
\newcommand{\DTIME}[1]{\text{\sc DTime}(#1)}
\newcommand{\FDTIME}[1]{\text{\sc FDTime}(#1)}
\newcommand{\cval}[1]{\N^{{#1}}}

\newcommand{\prob}[3]{
\begin{quote}
{\sc #1}
\begin{description}
\item[Input] #2.
\item[Question] #3?
\end{description}
\end{quote}
}

\newcommand{\vas}{{\sc VAS}\xspace}
\newcommand{\vass}{{\sc VASS}\xspace}

\begin{abstract}
Petri nets, equivalently presentable as vector addition systems with states, 
are an established model of concurrency with widespread applications.
The reachability problem, where we ask whether from a given initial configuration 
there exists a sequence of valid execution steps reaching a given final configuration,
 is the central algorithmic problem for this model.
The complexity of the problem has remained, until recently, one of the hardest open questions 
in verification of concurrent systems.  
A first upper bound has been provided only in 2015 by Leroux and Schmitz, then refined by the same authors
to non-primitive recursive Ackermannian upper bound in 2019.
The exponential space lower bound, shown by Lipton already in 1976,
remained the only known for over 40 years until a breakthrough non-elementary lower bound
by  Czerwi{\'n}ski, Lasota, Lazic, Leroux and Mazowiecki in 2019.
Finally, a matching Ackermannian lower bound announced this year by Czerwi{\'n}ski and Orlikowski, and
independently by Leroux, established the complexity of the problem.

Our primary contribution is  an improvement of the former construction, making it conceptually simpler and more direct.
On the way we improve the lower bound for vector addition systems with states in fixed dimension (or, equivalently, 
Petri nets with fixed number of places):
while Czerwi{\'n}ski and Orlikowski prove $\FF k$-hardness (hardness for $k$th level in Grzegorczyk Hierarchy)
in dimension $6k$,
our simplified construction yields $\FF k$-hardness already in dimension $3k+2$.
%
\end{abstract}

\newpage


\section{Introduction}

Petri nets~\cite{Petri62} are an established model of concurrency
with extensive and diverse applications in various fields, including
 modelling and analysis of hardware \cite{BurnsKY00,LerouxAG15}, software \cite{GermanS92,BouajjaniE13,KKW14} 
and database \cite{BojanczykDMSS11} systems, as well as chemical~\cite{AngeliLS11}, biological \cite{BaldanCMS10} and business \cite{Aalst15,LiDV17} processes (the references on applications are 
illustrative).  
The model admits various alternative but essentially equivalent presentations, most notably 
\emph{vector addition systems} (\vas) \cite{KarpM69}, and 
\emph{vector addition systems with states} (\vass) \cite[Sec.5]{Greibach78a}, \cite{HopcroftP79}.
The central algorithmic question for this model is the \emph{reachability problem} that 
asks whether from a given initial configuration there exists a sequence of valid execution steps reaching a given final configuration.
Each of the alternative presentations admits its own formulation of the reachability problem, all of them being equivalent  
due to straightforward polynomial-time translations that preserve reachability, see 
e.g.\ Schmitz's survey \cite[Section~2.1]{Schmitz16siglog}.
For instance, in terms of \vas, the problem is stated as follows:
given a finite set~$T$ of integer vectors in $d$-dimensional space and two $d$-dimensional vectors $\mathbf{v}$ and 
$\mathbf{w}$ of nonnegative integers, does there exist a walk from $\mathbf{v}$ to $\mathbf{w}$ such that 
it stays within the nonnegative orthant, and every step modifies the current position by adding some vector from~$T$?
The model of \vass is a natural extension of \vas with finite control, where $\mathbf{v}$ is additionally equipped with an initial control state,
$\mathbf{w}$ with a final one, and each vector in $T$ is additionally equipped with a source-target pair of control states.

We recall, following~\cite{Schmitz16siglog,stoc19,jacm,CLO21}, that importance of the Petri nets reachability 
problem is widespread, as many diverse 
problems from formal languages
\cite{Crespi-ReghizziM77}, logic \cite{Kanovich95,DemriFP16,DeckerHLT14,ColcombetM14}, 
concurrent systems~\cite{GantyM12,EsparzaGLM17}, 
process calculi~\cite{Meyer09}, linear algebra~\cite{HL18} and other areas
(the references are again illustrative) are known to admit reductions from the \vass reachability problem;
for more such problems and a wider discussion, we refer to~\cite{Schmitz16siglog}.

\para{Brief history of the problem}

The complexity of the Petri nets reachability problem has remained unsettled over the past half century.
Concerning the decidability status,
after an incomplete proof by Sacerdote and Tenney in 1970s~\cite{SacerdoteT77}, decidability of the
problem was established by Mayr~\cite{Mayr81,Mayr84} in 1981, whose proof was
then simplified by Kosaraju~\cite{Kosaraju82}, and then further refined by Lambert in the 1990s~\cite{Lambert92}.
A different approach, based on inductive invariants, has emerged from
a series of papers by Leroux a decade ago \cite{Leroux10,Leroux11,Leroux12}.

Concerning upper complexity bounds, the first such bound has been shown only in 2015 by Leroux and Schmitz~\cite{LerouxS15}, 
consequently improved to the Ackermannian upper bound~\cite{LS}. 

Concerning lower complexity bounds,
Lipton's landmark exponential space lower bound from 70ies~\cite{lipton76} 
has remained the state of the art for over 40 years until 
a breakthrough non-elementary lower bound by Czerwi{\'n}ski, Lasota, Lazic, Leroux and Mazowiecki in 2019~\cite{stoc19}
(see also~\cite{jacm}): hardness of the reachability problem for the class \tower
of all decision problems that are solvable in time or space bounded by a tower of exponentials whose height
is an elementary function of input size.
A further refinement of \tower-hardness,
in terms of fine-grained complexity classes closed under polynomial-time reductions,
has been reported by Czerwi{\'n}ski, Lasota and Orlikowski~\cite{CLO21}.
Finally, a matching Ackermannian lower bound has been announced recently, independently by
Czerwi{\'n}ski and Orlikowski~\cite{CO}, and by Leroux~\cite{L}
(the two constructions underlying the proofs seem to be significantly different).
These results finally close the long standing complexity gap, and yield \ackermann-completness of the 
Petri nets reachability problem.
The techniques used in~\cite{CO} and~\cite{L} substantially differ.

\para{Our contribution}
We provide  an improvement of the construction of~\cite{CO}.
As our main contribution, we make the construction conceptually simpler and more direct
(the idea of improvement is discussed at the end of Section~\ref{sec:reach}, and the central ingredient
of our construction is presented in Section~\ref{sec:transf}).
Moreover, on the way we improve the parametric lower bound with respect to the dimension of 
vector addition systems with states 
(or, equivalently, the number of places of Petri nets%
\footnote{We remark that a Petri net corresponding to a \vass of dimension $d$ has $d+3$ places,
due to $3$ extra places for encoding the control states of \vass~\cite{HopcroftP79}.
Likewise, a \vas corresponding to a \vass of dimension $d$ has dimension $d+3$.}).
For formulating the result we refer to the complexity classes $\FF \alpha$ corresponding to 
the Grzegorczyk hierarchy of fast-growing functions~\cite{lob1970hierarchies,Schmitz16toct}, 
indexed by ordinals $\alpha = 0, 1, 2, \ldots, \omega$; for instance, the class $\FF 3$ is \tower 
(class of all decision problems that are solvable in time or space bounded by a tower of exponentials, 
closed under elementary reductions) 
and $\FF \omega$ is \ackermann
(class of all decision problems that are solvable in time or space bounded by the Ackermann function,
closed under primitive-recursive reductions).
Results of~\cite{CO,L} can be stated in parametric terms as follows:
the former shows $\FF k$-hardness of the reachability problem for \vass in dimension $6k$, 
while  the latter one shows $\FF k$-hardness for \vass in dimension $4k+5$.
Our simplified construction yields a better lower bound: $\FF k$-hardness already in dimension $3k+2$.
This improvement is a step towards establishing the tight dimension-parametric complexity of the problem,
as the best known upper bound is $\FF k$-membership in dimension $k-4$~\cite{LS}.
As a next step, an improvement of the construction of~\cite{L} to dimension $2k+4$ 
has been recently reported in~\cite{L-arxiv}.

\section{The reachability problem}
\label{sec:reach}

In this section we define the reachability problem and explain our contribution.
Following~\cite{stoc19,jacm,CLO21,CO,L},
we work with a convenient presentation of \vass as counter programs without zero tests,
where the dimension of a \vass corresponds to the number of counters of a program.

\para{Counter programs}
A \emph{counter program} (or simply a \emph{program}) is a sequence of (line-numbered) commands, each of which is 
of one of the following types:
\begin{quote}
\begin{tabular}{l@{\qquad}l}
\inc{\vr{x}}                       & (increment counter~$\vr{x}$) \\
\dec{\vr{x}}                       & (decrement counter~$\vr{x}$) \\
\goto{$L$}{$L'$}                   & (nondeterministically jump to either line~$L$ or line~$L'$) \\
\testz{\vr{x}}                     & (zero test: continue if counter~$\vr{x}$ equals $0$)
\end{tabular}
\end{quote}
%
%
Counters are only allowed to have nonnegative values.
We are particularly interested in counter programs \emph{without zero tests}, i.e., ones that use no zero test command.
Whenever we use zero tests in the sequel, it is always in view of faithfully simulating them by programs without zero tests.

\smallskip

\begin{quote}
\underline{Convention:}
In the sequel, unless specified explicitly, counter programs are implicitly assumed to be \underline{without zero tests}.
\end{quote}

%

\begin{example} \label{ex:1}
We write \add{\vr x}{m} (resp.~\sub{\vr x}{m}) as a shorthand for for $m$ consecutive increments (resp.~decrements) of $\vr x$.
As an illustration, consider the program with three counters $\vr C = \set{\vr x, \vr y, \vr z}$ (on the left),
and its more readable presentation using a syntactic sugar \textbf{loop} (on the right):

\PROGnoname{0.4}{ex}{
\State \goto 2 6
\State \dec{\vr{x}}  \label{l:testz.iter} \label{l:testz.iter2}
\State \inc{\vr{y}}                
\State \add{\vr{z}}{2}                               \label{l:testz.end}
\State \goto 1 1
\State \inc{\vr z}
}
\PROGnoname{0.4}{loop}{
\Loop
\State \dec{\vr{x}}  
\State \inc{\vr{y}}                
\State \add{\vr{z}}{2}                              
\EndLoop
\State \inc{\vr z}
\vspace{0.4cm}
}
\medskip

\noindent
The program repeats the block of commands in lines \ref{l:testz.iter}--\ref{l:testz.end} 
some number of times chosen nondeterministically (possibly zero, although not infinite because $\vr{x}$ is decreasing,
and hence its initial value bounds the number of iterations) and then increments $\vr z$.
In the sequel we conveniently use \textbf{loop} construct instead of explicit \textbf{goto} commands.
(A dummy command is implicitly added after a \textbf{loop} in case it appears at the very end of a program.)

We emphasise that counters are only permitted to have nonnegative values.  
In the program above, that is why the decrement in line~\ref{l:testz.iter2} works also as a non-zero test.
\end{example}


Consider a program with counters $\vr C$.
By $\cval{\vr C}$
we denote the set of all valuations of counters.
Given an initial valuation of counters,
a \emph{run} (or \emph{execution}) of a counter program is a finite 
sequence of executions of commands, as expected.
A run which has successfully finished 
we call \emph{complete};
otherwise, the run is \emph{partial}.  
Observe that, due to a decrement that would cause a counter to become negative, 
a partial run may fail to continue because it is blocked from further execution.  
Moreover, due to nondeterminism of \textbf{goto}, a program may have various runs
from the same initial valuation.

Two programs $\prog P, \prog Q$ may be \emph{composed} by concatenating them, written $\prog P\ \prog Q$.
We silently assume the appropriate re-numbering of lines referred to by \textbf{goto} command in $\prog Q$.

\para{The reachability problem}

Given a subset $R\subseteq \cval{\vr C}$ of valuations, by a run \emph{from $R$} we mean any run 
whose initial valuation belongs to $R$.
A complete run  is called \emph{$\vr X$-zeroing}, for a subset $\vr X \subseteq \vr C$ of counters, 
if it ends with $\vr x = 0$ for all $\vr x \in \vr X$.
When $\vr X = \set{\vr x}$ and/or $R = \set{r}$ are a singleton we write simply ``$\vr x$-zeroing'' and/or ``from $r$''.
%
For instance, the program from Example~\ref{ex:1} has exactly one $\vr x$-zeroing run from the 
valuation $\vr x = 10$, $\vr y = \vr z = 0$, where the final values of counters are $\vr x = 0$, $\vr y = 10$, $\vr z = 21$.

By $\zeroval$ we denote the valuation where all counters are 0.
Following~\cite{stoc19,jacm,CLO21,CO,L}, we investigate the complexity of the following variant of the reachability problem
(with a partially specified final valuation of counters):

\prob{Reachability problem}
{A program $\prog P$ without zero tests, and a subset $\vr X$ of its counters}
{Does $\prog P$ have an $\vr X$-zeroing run from the zero valuation $\zeroval$}

Since counter programs without zero tests can be seen as presentations of \vass, 
the above decision problem translates to a variant of the reachability problem for the latter model, 
where all components of the initial vector are 0, and 
the specified components of the final vector are required to be 0. 
This variant polynomially reduces to the classical one where all components of the final vector are 
fully specified (say, required to be 0),
and the reduction preserves dimension.
According to the encoding of \vass as Petri nets, the problem 
translates to the \emph{submarking reachability} problem for the latter model,
where all places (except for those encoding the control states) are initially empty,
and the specified places are required to be finally empty.
Finally, the submarking reachability problem is polynomially equivalent to a variant where the final content of all places
is fully specified.

\para{Fast-growing hierarchy} 
For a positive integer $k$, let $\Npar k = \set{k, 2k, 3k, \ldots} \subseteq \N$ denote positive multiplicities of $k$.
We define the complexity classes $\FF i$ corresponding to the $i$th level in the Grzegorczyk 
Hierarchy~\cite[Sect.~2.3, 4.1]{Schmitz16toct}.
The standard family of approximations $\Ack i :\Npar 1 \to \Npar 1$ of Ackermann function,
for $i \in \Npar 1$, can be defined as follows:
 \[
\Ack 1(n) = 2n, \qquad \Ack {i+1}(n) = \underbrace{\Ack i \circ \Ack i \circ \ldots \circ \Ack i}_{n}(1) = 
\Ack i^n(1).
\]
%
In particular, $\Ack 2(n) = 2^n$ and $\Ack i(1) = 2$ for all $i\in\Npar 1$. 
Using functions $\Ack i$, we define  the complexity classes $\FF i$, indexed by $i\in \Npar 1$, of problems solvable in deterministic time
$\Ack i(p(n))$, where $p : \Npar 1\to\Npar 1$ ranges over functions computable in deterministic time $\Ack {i-1}^m(n)$, for
some $m\in\Npar 1$:
\[
\FF i = \bigcup_{p\in\FFfun {i-1}} \DTIME{\Ack i(p(n))}, \qquad \text{ where }
\FFfun i = \bigcup_{m\in\Npar 1} \FDTIME{\Ack {i}^m(n)}.
\]
Intuitively speaking, the class $\FF i$ contains all problems solvable in time $\Ack i(n)$, and is closed under reductions
computable in time of lower order $\Ack {i-1}^m(n)$, for some fixed $m\in\Npar 1$.
In particular, $\FF 3 = \tower$ (problems solvable in a tower of exponentials of time or space, 
whose height is an elementary function of input size).
The classes $\FF k$ are robust with respect to
the choice of fast-growing function hierarchy (see \cite[Sect.4.1]{Schmitz16toct}).
For $k\geq 3$,  instead of deterministic time, one could equivalently take nondeterministic time, or space.

\para{Dimension-parametric lower bound}
As the main result we prove $\FF k$-hardness for programs with the fixed number $3k+2$ of counters:
\begin{theorem}\label{thm:reach}
Let $k\geq 3$.
The  reachabilty problem for programs with $3k+2$ counters is $\FF k$-hard.
\end{theorem}
The proof is in Section~\ref{sec:proofreach}.
The result can be compared to $\FF k$-hardness shown in \cite{CO} for $6k$ counters, and
in \cite{L} for $4k+5$ counters.
Like the cited results, Theorem~\ref{thm:reach} implies \ackermann-hardness for unrestricted number of counters
which, together with \ackermann upper bound of \cite{LS}, yields \ackermann-completness
of the reachability problem.

\para{Idea of simplification}

Czerwi{\'n}ski and Orlikowski~\cite{CO} use the \emph{ratio technique} introduced previously in~\cite{stoc19}.
Speaking slightly informally, suppose some three counters $\vr b, \vr c, \vr d$ satisfy initially
\begin{align} \label{eq:ratio}
\vr b = B, \qquad  \vr c > 0, \qquad \vr d = \vr b \cdot \vr c, 
\end{align}
for some fixed positive integer $B\in\N$.
Furthermore, suppose that the initial values of $\vr c$ and $\vr d$ may be 
arbitrary, in a nondeterministic way, as long as they satisfy the latter equality in~\eqref{eq:ratio}; 
they are hence unbounded.
Under these assumptions, the ratio technique of~\cite{stoc19} allows one to correctly simulate unboundedly many zero tests 
(in fact, the number of simulated zero tests corresponds to the initial value of $\vr c$ which may be arbitrarily large)
on counters bounded by $B$, at the price of using some auxiliary counters.

As our technical contribution, we improve and simplify the ratio technique.
The core idea underlying our  simplification is, intuitively speaking, 
to swap the roles of counters $\vr b$ and $\vr c$:
we observe that the above-defined assumption~\eqref{eq:ratio} 
allows us to correctly simulate exactly $B/2$ zero tests (for $B$ even) on unbounded counters
(in fact, on counters bounded by the initial value of $\vr c$ which may be arbitrarily large), without any auxiliary counters.
This novel approach is presented in detail in Section~\ref{sec:transf}.

\section{Multipliers} \label{sec:mult}

Following the lines of~\cite{stoc19,jacm,CO}, we rely on a concept of \emph{multiplier}.

\vspace{-3mm}

\para{Sets computed by programs}
Consider a program $\prog P$ with counters $\vr C$, 
a set of counters $\vr X \subseteq \vr C$ and $R\subseteq \cval{\vr C}$.
We define the set \emph{$\vr X$-computed by $\prog P$  from $R$} 
as the set of all valuations of counters at the end of all $\vr X$-zeroing (and hence forcedly complete)
runs of $\prog P$ from $R$.
Formally, denoting by $\runs {\prog P} R {\vr X}$ the set of all $\vr X$-zeroing runs of $\prog P$ from $R$,
and by $\finval \pi$ the final counter valuation of a complete run $\pi$ of $\prog P$, 
the set $\vr X$-computed by $\prog P$ from $R$ is
\[
\computed {\prog P} R {\vr X} \ = \  \setof{\finval \pi}{\pi \in \runs {\prog P} R {\vr X}}.
\]
%
%
We omit $\vr X$ when it is irrelevant.
As before, when $\vr X = \set{\vr x}$ and/or $R = \set{r}$ are a singleton we write simply '$\vr x$-computed' and/or 'from $r$'.

\begin{example}
The program in Example~\ref{ex:1} above, $\vr x$-computes from the set of all valuations satisfying $\vr y = \vr z = 0$
(no constraint for $\vr x$),
the set of all valuations satisfying $\vr x = 0$ (trivially) and $\vr z = 2\vr y+1$.
\end{example}

Likewise, for a fixed integer $m\in\N$ and a program $\prog P$ with zero tests, 
we define the set $\vr X$-computed by $\prog P$ from $R$ 
\emph{using $m$ zero tests}, by restricting the above definition to runs $\pi\in\runs {\prog P} R {\vr X}$ that do exactly $m$ zero tests. 
This finer variant of the definition will be used in the next section.

\para{Multipliers}
Let $\vr b, \vr c, \vr d \in \vr C$ be some  three distinguished counters, and $B\in\Ne$.
We define the subset $\ratio{\succeq}(B, \vr b, \vr c, \vr d, \vr C)\subseteq \cval {\vr C}$,
called informally the \emph{ratio of $B$}, consisting of all valuations that satisfy the three conditions~\eqref{eq:ratio}
and assign 0 to all other counters $\vr x\in\vr C \setminus \set{\vr b, \vr c, \vr d}$.
%
%
\begin{definition} \label{def:mult}
A program $\prog M$ (with no zero tests) 
with counters $\vr C$ that $\vr z$-computes from the zero valuation $\zeroval$  the set $\ratio{\geq}(B, \vr b, \vr c, \vr d, \vr C)$,
for some four of its counters $\vr z, \vr b, \vr c, \vr d\in\vr C$,
we call {\bf $B$-multiplier}.
In formula: $\computed {\prog M} {\zeroval} {\vr z} = \ratio{\geq}(B, \vr b, \vr c, \vr d, \vr C)$.
\end{definition}

\begin{example} \label{ex:mult}
As a simple example, for every fixed $B\in\Ne$, the following program is a $B$-multiplier of size $\O(B)$
(several commands are written in one line to save space).
Counter $\vr z$ is not used at all.

\PROG{0.55}{Program $\mult B(\vr b, \vr c, \vr d)$}{prog:mult}{
\State \add{\vr b} B
  \quad \add{\vr d} B \quad \inc{\vr c} 
\Loop
  \State \add{\vr d} B \quad \inc{\vr c} 
\EndLoop
}
\end{example}

\medskip
%
%



\noindent
Directly from the definition we derive the following fundamental property of multipliers, to be used in the proofs  
in Sections~\ref{sec:proofmult} and \ref{sec:proofreach}:
\begin{claim} \label{claim:mult}
Let $\prog M$ be a $B$-multiplier with counters $\vr C$ as in Definition~\ref{def:mult}, 
let $\prog P$ be a counter program 
with counters $\vr C \setminus\set{\vr z}$,
and let $\vr Y\subseteq \vr C$.
Then the set $\vr Y$-computed by $\prog P$ from $\ratio{\geq}(B, \vr b, \vr c, \vr d, \vr C)$ is equal to
the set $(\set{\vr z} \cup \vr Y)$-computed by the composed program $\prog M \ \prog P$ from $\zeroval$:
\[
\computed {\prog P} {\ratio{\geq}(B, \vr b, \vr c, \vr d, \vr C)} {\vr Y} \ = \ 
\computed {\prog M \, \prog P} {\zeroval} {\set{\vr z} \cup \vr Y}.
\]
%
\end{claim}
\begin{proof}
The claim is a special case of the following general composition rule:
for two programs $\prog P$ and $\prog Q$,
if $\computed {\prog P} A {\vr X} = B$
and $\prog Q$ does not use counters $\vr X$,
then $\computed {\prog P\, \prog Q} A {\vr X \cup \vr Y} = \computed {\prog Q} B {\vr Y}$.
Indeed, under the above assumptions 
$(\vr X \cup \vr Y)$-zeroing runs of $\prog P \ \prog Q$ from $A$ are in mutual correspondence with
$\vr Y$-zeroing runs of $\prog Q$ from $B$.
\end{proof}

\para{Computing multipliers}

For technical convenience we prefer to rely on the following family of functions $\F i : \Ne\to\Ne$, indexed by $i\in \Npar 1$,
closely related to functions $\Ack i$ (cf.~Claim~\ref{claim:ack} below):
\begin{align} \label{eq:F1}
\F 1(n) = 2n, \qquad \F {i+1} = \successor{{\F i}} \quad \text{where } \quad
\successor{F}(n) = \underbrace{F \circ F  \circ \ldots \circ F}_{n / 4}(4).
\end{align}
%
By induction on $i$ one easily shows that $\F i$ is a linear re-scaling of $\Ack i$:
%
\begin{claim} \label{claim:ack}
$\F i(4\cdot n) = 4 \cdot \Ack i(n)$, for $i, n\in \Npar 1$.
\end{claim}
\begin{proof}
As $\F 1(n) = 2n$ and $\Ack 1(n) = 2n$, the claim holds for $i = 1$.
Assuming the claim for $i\in\Npar 1$, by $n$-fold application thereof we derive
the required equality for $i+1$:
\[
\F {i+1}(4\cdot n) \ = \  
\underbrace{\F  i \circ  \ldots \circ \F i}_{n}(4) \ = \ 
4 \cdot \underbrace{\Ack i \circ \ldots \circ \Ack i}_{n}(1) \ = \ 
4 \cdot \Ack {i+1}(n).
\]
\end{proof}

As a technical core of the proof of Theorem~\ref{thm:reach}, combining our simplification with the lines of \cite{CO},
we provide an effective construction of $B$-multipliers with $3k+2$ counters, where $B = \F k(n)$, of size polynomial in
$k$ and $n$.

\begin{theorem} \label{thm:multrelaxed}
Given $k\in\Npar 1$ and $n\in\Ne$ one can compute, in time polynomial in $k$ and $n$,
an $\F k(n)$-multiplier with $3k+2$ counters.
\end{theorem}
The proof is in Section~\ref{sec:proofmult}.

%
%
%

\section{Bounded number of zero tests} \label{sec:transf} 

In this section we provide a novel construction that enables simulating a bounded number $m$ of 
zero tests (cf.~Lemma~\ref{lem:halt}) at the cost of introducing additional 3 counters initialised to
the ratio of $B = 2(m+1)$.
This construction is a core ingredient of the proofs of Theorems~\ref{thm:reach} and~\ref{thm:multrelaxed}.

Whenever analysing a single run of a program, we denote by $\ini x$ the initial value of a counter $\vr x$,
and by $\fin x$ the final value thereof.

\para{Maximal iteration}  

In the sequel we intensively use loops of the following form that, intuitively, flush the value from counter 
$\vr f$ to $\vr e$, 
decreasing simultaneously counter $\vr d$ (and possibly execute some further commands):

\begin{align}
\PROGnoname{0.5}{prog:flush}{
\Loop
  \State \dec{\vr f} \quad \inc{\vr e} \quad \dec{\vr d} \quad \ldots
\EndLoop
}
\end{align}

\bigskip

\noindent
Assuming $\ini d \geq \ini f$,
we observe that the amount $\ini d - \fin d$ by which $\vr d$ is decreased as an effect of execution 
(we use the word \emph{execution} as a synonym to \emph{complete run}) of the above loop
may be any value between $0$ and $\ini f$.
Furthermore, assuming $\ini d \geq \ini e + \ini f$,
the equality
$\ini d - \fin d = \ini e + \ini f$ holds if and only if 
\begin{align} \label{eq:max}
\ini e = 0 = \fin f.
\end{align}
This simple observation will play a crucial role in the sequel, and deserves a definition:
\begin{definition}
Whenever an execution of a loop of the form~\eqref{prog:flush} satisfies the two equalities~\eqref{eq:max} we call this execution {\bf maximally iterated}.
\end{definition}
%

\para{The construction}

Let $\prog P$ be a counter program with counters $\vr C$, and assume that $\prog P$ uses zero tests only
on two its counters $\vr x, \vr y\in\vr C$ (the construction 
easily extends to programs with an arbitrary number of zero-tested counters). 
We add to $\prog P$ three fresh counters $\vr b, \vr c, \vr d$ 
(let $\vr {\transf C} = \vr C \cup \set{\vr b, \vr c, \vr d}$), 
and transform $\prog P$ into a program $\transf {\prog P}$ \emph{without zero tests} that,
assuming its initial valuation of counters belongs to $\ratio{=}(2(m+1), \vr b, \vr c, \vr d, \vr {\transf C})$
%
%
for some $m\in\N$, 
simulates correctly $m$ zero tests (jointly) on counters $\vr x, \vr y$,
as long as their sum is bounded by the initial value of $\vr c$
(cf.~Lemma~\ref{lem:halt}). 

The transformation proceeds in three steps.
First, we accompany every increment (decrement) on $\vr x$ with a decrement (increment) of $\vr c$,
and likewise we do for $\vr y$:

\begin{quote}
\begin{tabular}{l|c}
command & replaced by \\
\hline
\inc{\vr x} & \inc{\vr x} \quad \dec{\vr c} \\
\dec{\vr x} & \dec{\vr x} \quad \inc{\vr c} \\
\end{tabular}
\qquad\qquad\qquad
\begin{tabular}{l|c}
command & replaced by \\
\hline
\inc{\vr y} & \inc{\vr y} \quad \dec{\vr c} \\
\dec{\vr y} & \dec{\vr y} \quad \inc{\vr c} \\
\end{tabular}
\end{quote}

\noindent
In the resulting program $\prog{\overline P}$ counters $\vr x$, $\vr y$ are, intuitively speaking, put on
a shared 'budget' $\vr c$.
Assuming $\vr x$ and $\vr y$ initially 0, 
this clearly enforces $\vr x + \vr y$ to not exceed the initial value of $\vr c$, and the sum
$s = \vr c + \vr x + \vr y$ to remain invariant. 

As the second step, we replace in $\prog{\overline P}$
every \testz{\vr x}  command by the following macro:

\PROG{0.45}{$\zerotestnew {\vr x}{\vr y}{\vr c}{\vr d}$}{prog:zeroflushx}{
\Loop
  \State \dec{\vr y} \quad \inc{\vr x} \quad \dec{\vr d}
\EndLoop
\Loop
  \State \dec{\vr c} \quad \inc{\vr y} \quad \dec{\vr d}
\EndLoop
\Loop
  \State \dec{\vr y} \quad \inc{\vr c} \quad \dec{\vr d}
\EndLoop
\Loop
  \State \dec{\vr x} \quad \inc{\vr y} \quad \dec{\vr d}
\EndLoop
\State \sub{\vr b} 2
}

\bigskip

\noindent
Likewise we replace every \testz{\vr y} command by an analogous macro
$\zerotestnew {\vr y}{\vr x}{\vr c}{\vr d}$ obtained from $\zerotestnew {\vr x}{\vr y}{\vr c}{\vr d}$
by swapping $\vr x$ and $\vr y$.
This yields the program $\prog{\overoverline P}$ without zero tests. 
We note that each of the two $\zerotestname$ macros preserves the sum $s = \vr c + \vr x + \vr y$, and 
decrements the counter $\vr b$ by $2$.
Furthermore, each of the two  $\zerotestname$ macros  decrements $\vr d$  
by at most $2s$  (cf.~Claim~\ref{claim:zerotest} below), and hence the macros preserve
the inequality $d \geq \vr b \cdot s$ (recall that $\vr d = \vr b \cdot s$ holds initially). 
%

As the final step we adjoint at the end of $\prog{\overoverline P}$
the following program $\settozero$, 
thus obtaining the transformed program $\transf{\prog P}$:    

\PROG{0.4}{$\settozero$}{prog:zeroflushdc}{
\Loop
  \State \dec{\vr c} \quad \sub{\vr d} 2
\EndLoop
\State $\zerotestnew {\vr c}{}{}{}$
}


\bigskip

\noindent
The macro $\zerotestnew {\vr c}{\vr y}{\vr x}{\vr d}$ 
is obtained from $\zerotestnew {\vr x}{\vr y}{\vr c}{\vr d}$ by swapping $\vr x$ and $\vr c$.
We note that an execution of $\settozero$ may decrease the sum $\vr c + \vr x + \vr y$ (but
$\zerotestnew {\vr c}{}{}{}$ preserves it).

\para{Correctness}

Recall that $\prog{\overoverline P}$ preserves the sum $\vr c + \vr x + \vr y$;
we denote by $s$ the value of this sum.
An execution of $\zerotestnew {\vr x}{\vr y}{\vr c}{\vr d}$
is called \emph{maximally iterated} if all four loops are so.
Observe that every such execution is forcedly \emph{correct}, i.e.~satisfies:
\begin{align} \label{eq:correct}
 \ini x = \fin x = 0, \qquad  \ini y = \fin y, \qquad \ini c = \fin c.
\end{align}
(Likewise in case of $\zerotestnew {\vr y}{\vr x}{\vr c}{\vr d}$ and $\zerotestnew {\vr c}{\vr y}{\vr x}{\vr d}$.)
%
The idea behind $\zerotestnew {\vr x}{\vr y}{\vr c}{\vr d}$ is to flush from $\vr y$ to a zero-tested counter $\vr x$
and back, but also flush from $\vr c$ to $\vr y$ and back, in an appropriately nested way that guarantees that
the amount $\ini d - \fin d$ by which $\vr d$ is decreased equals $2s$ exactly in maximally iterated executions:
\begin{claim} \label{claim:zerotest}
Consider an  execution of $\zerotestnew {\vr x}{\vr y}{\vr c}{\vr d}$ (resp.~$\zerotestnew {\vr y}{\vr x}{\vr c}{\vr d}$) macro, 
assuming $\ini d \geq 2 s$.
Then $0 \leq \ini d - \fin d \leq 2s$. 
Furthermore, the equality $\ini d - \fin d = 2s$ holds if and only if the execution is maximally iterated. 
\end{claim}
\begin{claimproof}
Consider an execution of $\zerotestnew {\vr x}{\vr y}{\vr c}{\vr d}$, assuming $\ini d \geq 2 s$, and
let $\midd y$ denote the value of $\vr y$ at the exit from the first loop.
The amount by which $\vr d$ is decreased in the two loops in lines 1--2 and 7--8 
is at most
\[
\Delta_1 = 2(\ini  y - \midd y) + \ini x.
\]
Furthermore, the amount by which $\vr d$ is decreased
in the two loops in lines 3--6 is at most
\[
\Delta_2 = 2 \ini c + \midd y.
\]
The sum  $\Delta_1 + \Delta_2$ clearly satisfies  
$\Delta_1 + \Delta_2 \leq 2s = 2(\ini c + \ini x + \ini y)$.
It equals $2s$ if and only if $\Delta_1 = 2\ini y$ and $\Delta_2 = 2\ini c$, i.e., 
exactly when all four loops are maximally iterated.
\end{claimproof}
In consequence, as $\vr b$ is decreased by $2$, if the invariant $\vr d = \vr b \cdot s$ is preserved by 
an execution of $\zerotestnew {\vr x}{}{}{}$ (resp.~$\zerotestnew {\vr y}{}{}{}$) then the zero test is forcedly correct.
Furthermore notice that once the invariant is violated, i.e., $\vr d > \vr b \cdot s$, due to the first part of Claim~\ref{claim:zerotest}
the invariant can not be recovered later.
These observations lead to the correctness claim stated in Lemma~\ref{lem:halt}.

In the proof of Lemma~\ref{lem:halt} we will also need the following corollary of Claim~\ref{claim:zerotest},
where $s$ denotes, as before, the sum $\vr c + \vr x + \vr y$ at the start of $\settozero(\vr c)$:
\begin{claim} \label{claim:settozeroo}
Consider an execution of $\settozero(\vr c)$, assuming $\ini d \geq 2 s$.
Then $0 \leq \ini d - \fin d \leq 2s$. 
Furthermore, the equality $\ini d - \fin d =2s$ holds if and only if
the $\zerotestnew {\vr c}{}{}{}$ macro is maximally iterated.
\end{claim}
\begin{claimproof}
Consider an execution of $\settozero(\vr c)$ and
denote by $\dot s$ the value of $\vr c + \vr x + \vr y$ just before entering $\zerotestnew {\vr c}{}{}{}$.
Thus $\vr d$ decrease by $2(s - \dot s)$ before entering $\zerotestnew {\vr c}{}{}{}$.
Moreover, due to Claim~\ref{claim:zerotest}, the macro $\zerotestnew {\vr c}{}{}{}$ decreases $\vr d$ by 
at most $2\dot s$, and furthermore the macro decreases $\vr d$ by exactly $2\dot s$ if
and only if it is maximally iterated.
These observations imply the claim.
\end{claimproof}

Recall that $\vr{\transf C} = \vr C  \cup \{\vr b, \vr c, \vr d\}$.
We define the \emph{$\vr{\transf C}$-extension} of a counter valuation $v\in\cval {\vr C}$ as the extension of
$v$ where $\vr b, \vr c$ and $\vr d$ are all set to $0$. 
The $\vr{\transf C}$-extension of a set $R \subseteq \cval{\vr C}$ is defined as the set of
$\vr{\transf C}$-extensions of all valuations in $R$.

\begin{lemma} \label{lem:halt}
The following sets are equal (as subsets of\, $\cval{\vr{\transf C}}$):
\begin{itemize}
\item the $\vr{\transf C}$-extension of the set computed by $\prog P$ from $\zeroval$ using $m$ zero tests.
\item the set $\vr d$-computed by $\transf{\prog P}$ from $\ratio{=}(2(m+1), \vr b, \vr c, \vr d, \vr {\transf C})$.
\end{itemize}
\end{lemma}

\begin{proof}
For the inclusion of the former set in the latter, we show that
for each complete run $\pi$ of $\prog P$ from $\zeroval$ that does $m$ zero tests on 
$\vr x, \vr y$,
there is a corresponding $\vr d$-zeroing run of $\transf{\prog P}$ from  
$\ratio{=}(2(m+1), \vr b, \vr c, \vr d, \vr {\transf C})$, 
for any initial value $\ini c$ at least as large as the maximal value of the sum $\vr x + \vr y$ along $\pi$. 
The run iterates maximally $\zerotestnew {\vr x}{}{}{}$ and $\zerotestnew {\vr y}{}{}{}$ macros,
decrements $\vr c$ to $0$ in line 2 in $\settozero(\vr c)$, and 
then iterates maximally $\zerotestnew {\vr c}{}{}{}$.
Thus the final counter valuation of the run is the $\vr{\transf C}$-extension of the final counter valuation of $\pi$.

For the converse direction, consider a $\vr d$-zeroing run $\pi$ of $\transf{\prog P}$\, from 
$\ratio{=}(2(m+1), \vr b, \vr c, \vr d, \vr {\transf C})$.
The initial counter valuation satisfies the equalities $\vr b = 2(m+1)$ and $\vr d = 2(m+1) \cdot s$.
Each execution of $\zerotestnew {\vr x}{}{}{}$ or $\zerotestnew {\vr y}{}{}{}$ or $\settozero(\vr c)$
decreases $\vr b$ by $2$, and $\vr d$ by at most $2 s$
(by the first part of Claims~\ref{claim:zerotest} and Claim~\ref{claim:settozeroo}).
Therefore, since $\vr b$ and $\vr d$ are not affected elsewhere and
$\fin d = 0$ finally, we deduce:

\begin{claim} \label{claim:2ss}
Each execution of $\zerotestnew {\vr x}{}{}{}$, $\zerotestnew {\vr y}{}{}{}$  or 
$\zerotestnew {\vr c}{}{}{}$  in $\pi$
decreases $\vr d$ by \emph{exactly} $2 s$.
\end{claim}
\begin{claim}  \label{claim:corrzero2}
There are exactly $m$ executions of $\zerotestnew {\vr x}{}{}{}$ or $\zerotestnew {\vr y}{}{}{}$ in $\pi$.
\end{claim}
\begin{claim} \label{claim:finb}
Finally, $\fin b = 0$.
\end{claim}

By Claim~\ref{claim:2ss} and the second part of Claim~\ref{claim:zerotest} we derive:
\begin{claim} \label{claim:corrzero1}
Each execution of $\zerotestnew {\vr x}{}{}{}$ in $\pi$ is correct, i.e. satisfies the equalities~\eqref{eq:correct}.
Likewise for $\zerotestnew {\vr y}{}{}{}$.
\end{claim}
Analogously, by Claim~\ref{claim:2ss} and the second part of Claim~\ref{claim:settozeroo} we derive:
\begin{claim} \label{claim:finc}
Finally, $\fin c = 0$.
\end{claim}

Due to Claims~\ref{claim:corrzero2} and~\ref{claim:corrzero1}, once we project away from $\pi$ 
the counters $\vr b, \vr c, \vr d$, we obtain
a complete run of $\prog P$ from $\zeroval$ that does exactly $m$ zero tests, as required.
Finally, due to Claims~\ref{claim:finb} and~\ref{claim:finc}, the $\vr{\transf C}$-extension of the final counter valuation 
of the obtained run is exactly the final counter valuation of $\pi$.
\end{proof}

\section{Computing a large multiplier (Proof of Theorem~\ref{thm:multrelaxed})} \label{sec:proofmult} 

The proof proceeds by combining the concept of amplifier lifting of~\cite{CO} with 
the program transformation of Section~\ref{sec:transf}.

\para{Amplifiers}
Let $F : \Ne \to \Ne$ be a monotone function satisfying $F(n)\geq n$ for $n\in\Ne$.
Informally speaking, an \emph{$F$-amplifier} is a program without zero tests that computes
the ratio of $F(B)$ from the ratio of $B$, for every $B\in\Ne$.

\begin{definition}
Consider a program $\prog P$ with counters $\vr C$ without zero tests and distinguished three \emph{input counters}
$\vr b, \vr c, \vr d\in \vr C$ and three \emph{output counters} $\vr b', \vr c', \vr d'\in \vr C$.
The program is called {\bf $F$-amplifier} if for every $B\in\Ne$, 
it $\vr d$-computes from $\ratio{=}(B, \vr b, \vr c, \vr d, \vr C)$ the set $\ratio{\geq}(F(B), \vr b', \vr c', \vr d', \vr C)$.
%
%
\end{definition}
We note that no condition is imposed on $\vr d$-zeroing runs from counter valuations not
belonging to any set $\ratio{\geq}(B, \vr b, \vr c, \vr d, \vr C)$.
%
As an example, consider the following program $\linampl \ell$, for $\ell \in \Npar 1$, with input counters
$\vr b, \vr c, \vr d$ and output counters $\vr b', \vr c', \vr d'$:

\PROG{0.9}{Program $\linampl \ell(\vr b, \vr c, \vr d, \vr b', \vr c', \vr d')$}{prog:firstampl}{
\Loop
  \Loop
    \State \dec{\vr c} \quad \inc{\vr c'} \quad \dec{\vr d} \quad \add{\vr d'} \ell
  \EndLoop
  \Loop
    \State \dec{\vr c'} \quad \inc{\vr c} \quad \dec{\vr d} \quad \add{\vr d'} \ell
  \EndLoop
  \State \sub{\vr b} 2 \quad \add{\vr b'} {2\ell}
\EndLoop
\Loop
  \State \sub{\vr c} 1 \quad \inc{\vr c'} \quad \sub{\vr d} 2 \quad \add{\vr d'} {2\ell}
\EndLoop
\State \sub{\vr b} 2 \quad \add{\vr b'} {2\ell} 
}

\bigskip

\begin{claim} \label{claim:F1}
The above program is an $L_\ell$-amplifier, where $L_\ell : \Ne \to \Ne = (x\mapsto \ell  \cdot x)$.
\end{claim}
\begin{claimproof}[Proof sketch]
Writing counter valuations as vectors $(\vr b, \vr c, \vr d, \vr b', \vr c', \vr d')$, 
one shows that
the program $\vr d$-computes, from the set containing just one counter valuation
$(B, c, d, 0, 0, 0)$, the set containing one counter valuation
$(0, 0, 0, \ell\cdot B, c, \ell \cdot d)$.
Indeed, as $\vr d = 0$ finally, each of the two inner loops in lines 2--5, as well as the last loop in lines 7--8,
is forcedly maximally iterated.
\end{claimproof}
%
%
%
Putting $\ell=1$ we get an identity-amplifier $\linampl 1(\vr b, \vr c, \vr d, \vr b', \vr c', \vr d')$.


\para{Amplifier lifting}
Recall the definition~\eqref{eq:F1} of functions $\F i$; in particular $\F 1 = L_2$.
%
Let $\prog P$ be a program with counters $\vr C$, without zero tests, with distinguished input counters
$\vr b_1, \vr c_1, \vr d_1 \in \vr C$ and output counters $\vr b_2, \vr c_2, \vr d_2 \in \vr C$.
We describe a transformation of the program $\prog P$ 
to a program $\successor{{\prog P}}$, also without zero tests,
such that
assuming that $\prog P$ is an $F$-amplifier for some function $F : \Ne \to \Ne$, 
the program $\successor{{\prog P}}$ is an $\successor F$-amplifier.
The program $\successor{{\prog P}}$ uses, except for the counters of $\prog P$, three fresh counters $\vr b, \vr c, \vr d$.
Thus counters of $\successor{{\prog P}}$ are $\vr{\transf C} = \vr C \cup \set{\vr b, \vr c, \vr d}$.
We let input counters of $\successor{{\prog P}}$ be $\vr b, \vr c, \vr d$, and its output counters be 
$\vr b_2, \vr c_2, \vr d_2$.

In the transformation we use 
%
the identity-amplifier $\linampl{} = \linampl 1(\vr b_2, \vr c_2, \vr d_2, \vr b_1, \vr c_1, \vr d_1)$ 
with input counters $\vr b_2, \vr c_2, \vr d_2$ and output counters $\vr b_1, \vr c_1, \vr d_1$,
and
the $4$-multiplier $\prog M = \mult 4 (\vr b_1, \vr c_1, \vr d_1)$ of Example~\ref{ex:mult}, both without zero tests.
%
%
%
%
%
%
For defining the program $\successor{{\prog P}}$ we apply 
the transformation of Section~\ref{sec:transf}  
(with counters $\vr d_1$ and $\vr d_2$ in place of $\vr x$ and $\vr y$)
to the following program $\prog Q$
built using $\prog P$, $\linampl{}$ and $\prog M$:

\PROG{0.45}{Program $\prog Q$}
{prog:amplzero}{
\State $\prog M$
\Loop
  \State $\prog P$
  \State \testz{\vr d_1}
  \State $\linampl{}$  
  \State \testz{\vr d_2}
\EndLoop
\State $\prog P$
\State \testz{\vr d_1}
\vspace{0.3cm}
}
\quad
\PROG{0.4}{Program $\successor {{\prog P}}$}{prog:ampl}{
\State $\prog{\overline M}$  
\Loop
  \State $\prog{\overline P}$
  \State $\zerotestnew {\vr d_1}{}{}{}$
  \State $\overline{\linampl{}}$  
  \State $\zerotestnew {\vr d_2}{}{}{}$
\EndLoop
\State $\prog {\overline P}$
\State $\zerotestnew {\vr d_1}{}{}{}$
\State $\settozero$
}

\bigskip

\noindent
Formally, $\successor{{\prog P}} = \transf{\prog Q}$.
Intuitively speaking, the program $\prog Q$ directly implements the computation of $\successor F$ according
to the definition: with $2\ell+1$ zero tests it computes, from $\zeroval$, the ratio of  $F^{\ell+1}(4)$.
Note that counters of $\prog Q$ are $\vr C$ while counters of $\successor{{\prog P}}$ are $\transf{C}$.
%
%
Lemma~\ref{lem:ampl-pres} states the crucial amplifier-lifting property of the program transformation $\prog P \mapsto \successor{{\prog P}}$.
\begin{lemma} \label{lem:ampl-pres}
If $\prog P$ is an $F$-amplifier, then $\successor{{\prog P}}$ is an $\successor F$-amplifier.
\end{lemma}
\begin{proof}
Let $\prog P$ be an $F$-amplifier.
Thus for every $B \in \Ne$, 
$\computed {\prog P} {\ratio{=}(B, \vr b_1, \vr c_1, \vr d_1, \vr C)} {\vr d_1} =  
\ratio{\geq}(F(B), \vr b_2, \vr c_2, \vr d_2, \vr C)$.
The program $\linampl{}$, being an identity-amplifier, $\vr d_2$-computes from $\ratio{=}(B, \vr b_2, \vr c_2, \vr d_2, \vr C)$ the set 
$\ratio{\geq}(B, \vr b_1, \vr c_1, \vr d_1, \vr C)$.
Let $B = 4(\ell+1) \in\Ne$ for an arbitrary $\ell\in\N$.
As $\prog P$ is an $F$-amplifier and $\linampl{}$ is an identity-amplifier, we deduce:
\begin{claim}
$\prog Q$ computes from $\zeroval$ using $2\ell+1$ zero tests the set 
$\ratio{\geq}(F^{\ell+1}(4), \vr b_2, \vr c_2, \vr d_2, \vr{C})$.
\end{claim}
\noindent
As $\successor{{\prog P}} = \transf{\prog Q}$, by Lemma~\ref{lem:halt} we deduce:
\begin{claim}
$\computed {\successor{{\prog P}}} {\ratio{\geq}(4(\ell +1), \vr b, \vr c, \vr d, \vr {{\transf C}})} {\vr d}
= \ratio{\geq}(F^{\ell+1}(4), \vr b_2, \vr c_2, \vr d_2, \vr{C})$.
\end{claim}
\noindent
As $B\in \Ne$ was chosen arbitrarily and $\successor{F}(B) = F^{\ell+1}(4)$, 
the last claim says that $\successor{{\prog P}}$ is an
$\successor{F}$-amplifier.
\end{proof}

\begin{remark} \label{rem:doubling}
The program $\prog P$ appears twice in the body of $\successor{{\prog P}}$. 
This doubling can be easily avoided by re-structuring the loop using explicit \textbf{goto} commands.
%
%
%
In this way, the size of $\successor{{\prog P}}$ becomes larger than the size of $\prog P$ only by a constant.
\end{remark}

\begin{proof}[Proof of Theorem~\ref{thm:multrelaxed}]
We rely on 
Lemma~\ref{lem:ampl-pres}.
Given $k\in \Npar 1$ and $n\in\Ne$ we compute, in time linear in $k$, 
the $\F k$-amplifier $\prog A_k$ with $3k+3$ counters $\vr C$, 
by $(k-1)$-fold application of the amplifier lifting transformation $\prog P \mapsto \successor{{\prog P}}$ described above, 
starting from the $\F 1$-amplifier $\linampl 2$ of Claim~\ref{claim:F1}.
The construction is linear in $k$ due to Remark~\ref{rem:doubling}.
Let $\vr b, \vr c, \vr d\in \vr C$ be input counters of $\prog A_k$.
Relying on Claim~\ref{claim:mult} in Section~\ref{sec:mult}, 
the $\F {k}(n)$-multiplier is obtained by pre-composing $\prog A_k$ with an $n$-multiplier
(e.g.~$\mult n(\vr b, \vr c, \vr d)$ from Section~\ref{sec:reach}) that outputs the set 
$\ratio{\geq}(n, \vr b, \vr c, \vr d, \vr C)$.
The whole construction is thus linear in $n$.

Finally we observe that the counter $\vr b$ is bounded by $n$ and hence can be eliminated:
we encode its values in control locations, by cloning the program into $n+1$ copies, where $i$th (for $i=0, \ldots, n$)
copy corresponds to the value $\vr b = i$.
%
%
The resulting program has $3k+2$ counters.
\end{proof}

\section{Hardness of the reachability problem (Proof of Theorem~\ref{thm:reach})} \label{sec:proofreach}

Relying on Lemma~\ref{lem:halt} and Theorem~\ref{thm:multrelaxed}, 
we prove in this section Theorem~\ref{thm:reach}.
%
%
Fix $k\geq 3$.
The proof proceeds by a polynomial-time reduction
from the following $\FF k$-hard problem:

\prob{$\F k$-bounded halting problem}
{A program $\prog P$ of size $n$ (w.l.o.g.~assume $n\in\Ne$) with 2 zero-tested counters}
{Does $\prog P$ have a complete run from $\zeroval$ that does at most $(\F k (n) -1)/ 2$ zero tests}
\vspace{-5mm}
\noindent
\begin{claim}
The above problem is $\FF k$-hard. 
\end{claim}
\begin{claimproof}
Indeed, the standard $\FF k$-hard halting problem
(does a program $\prog P$ with \emph{arbitrarily many} zero-tested counters $\vr x_1, \ldots, \vr x_\ell$ 
have a complete run that does at most $(\F k (n) -1)/ 2$  \emph{steps}?)
reduces polynomially to the above one using the standard simulation of arbitrarily many zero-tested counters 
by 2 such counters $\vr y_1, \vr y_2$.
The simulation stores the values of all counters $\vr x_1, \ldots, \vr x_\ell$ on one of $\vr y_1, \vr y_2$ 
(e.g., using G\"odel encoding), and the simulation of each command
involves flushing the value of that counter to the other, followed by the zero test.
Thus a bound on time of computation is translated to the same bound on the number of zero tests.
\end{claimproof}

Given  $\prog P$ as above with two counters $\vr x, \vr y$, 
we transform it to a counter program $\prog P'$ with $3k+2$ counters $\vr C$ but without zero tests, 
such that
$\prog P$ has a complete run from $\zeroval$ that does at most $m = (\F k (n) -1)/ 2$ zero tests if and only if
$\prog P'$ has a $\set{\vr d, \vr z}$-zeroing run from $\zeroval$ (for some $\vr d, \vr z\in\vr C$).

First,  we post-compose $\prog P$ with a simple program $\prog L$ that first decrements $\vr x$
nondeterministically many times, and then zero tests it
nondeterministically many times:

\medskip

\PROGnoname{0.5}{prog:zeroloop}{
\Loop
\State \dec {\vr x}
\EndLoop
\Loop
\State \testz {\vr x}
\EndLoop
}

\bigskip
\noindent
Thus $\prog P$ has a complete run that does \emph{at most} $m$ zero tests if and only if 
the composed program $\comp{\prog P}{\prog L}$
has a complete run that does \emph{exactly} $m$ zero tests.
We will apply the transformation of Section~\ref{sec:transf} to the composed program $\comp{\prog P}{\prog L}$.
Let $\vr b, \vr c, \vr d$ be the three counters added in the course of the transformation.

Second, using Theorem~\ref{thm:multrelaxed}
we compute 
a $2(m+1)$-multiplier $\prog M$ (recall that $2(m+1) = \F k(n)$) with $3k+2$ counters $\vr C$
that $\vr z$-computes from $\zeroval$ the set $\ratio{\geq}(2(m+1), \vr b, \vr c, \vr d, \vr C)$, for some counter $\vr z$ 
different than $\vr x, \vr y$.
Thus $\vr z, \vr b, \vr c, \vr d \in\vr C$.

Finally, we define $\prog P'$ as a composition of
$\prog M$ with the transformed program $\transf{(\comp{\prog P}{\prog L})}$,
and get the required equivalence:
\begin{claim} \label{claim:equiv}
The following conditions are equivalent:
\begin{itemize}
\item
$\prog P$ has a complete run from $\zeroval$ that does at most $m$ zero tests;
\item
$\comp{\prog P}{\prog L}$ has a complete run from $\zeroval$ that does exactly $m$ zero tests;
\item
$\transf{(\comp{\prog P}{\prog L})}$ has a $\vr d$-zeroing run from $\ratio{\geq}(2(m+1), \vr b, \vr c, \vr d, \vr {C})$;
\item
$\prog P' = \comp {\prog M} {(\transf{\comp{\prog P}{\prog L})}}$
has a $\set{\vr z, \vr d}$-zeroing run from $\zeroval$.
\end{itemize}
\end{claim}

\noindent
The second and the third point are equivalent due to Lemma~\ref{lem:halt}, 
while the equivalence of the third and the
last point follows by Claim~\ref{claim:mult} in Section~\ref{sec:mult}.

The program $\prog P'$ has $3k+4$ counters ($3k+2$ counters of $\prog M$ plus 
$\vr x, \vr y$) but,
since $k \geq 3$, this number can be decreased back to $3k+2$, by re-using some of $3k-2$  
counters from $\vr C' = \vr C - \set{\vr b, \vr c, \vr d, \vr z}$ in place of $\vr x, \vr y$.
The latter equivalence in Claim~\ref{claim:equiv} remains true, as $\vr z$-zeroing runs of $\prog M$ from $\zeroval$
are necessarily $\vr C'$-zeroing too, by the definition of multipliers. 

This completes the proof of Theorem~\ref{thm:reach}.

\section{Final remarks}

Primarily, we propose a conceptual simplification of the \ackermann-hardness construction of~\cite{CO}.

As a secondary achievement, we improve the dimension-parametric 
lower bound for the \vass (Petri nets) reachability problem:
compared to $\FF k$-hardness in dimension $6k$~\cite{CO} and $4k+5$~\cite{L}, respectively, we obtain
$\FF k$-hardness already in dimension $3k+2$.
(We believe that by combining with the insights of~\cite{CO} one can further optimise our construction and 
lower the dimension by a small constant.)
The dimension $4k+5$ of~\cite{L} has been recently further improved to $2k+4$~\cite{L-arxiv}, thus beating ours.

The best known upper bound places the \vass reachability problem in dimension $k-4$ is in $\FF k$~\cite{LS}.
Establishing exact parametric complexity of the problem, i.e., closing the gap between dimensions $k-4$ and $2k+4$, 
arises therefore as an intriguing open problem.

Finally we remind that except for dimension 1 and 2, where the reachability problem seems to be 
well understood~\cite{BlondinFGHM15,EnglertLT16},
we know no additional complexity bounds for small fixed dimensions $k$ except for the lower bound
derived from dimension 2, and the generic $\FF {k+4}$ upper bound of~\cite{LS}.

\newpage

\bibliography{bib}

\end{document}